\documentclass[runningheads]{llncs}
\pdfoutput=1
\usepackage{graphicx}
\usepackage{hyperref}

\usepackage{subcaption} 

\usepackage{amsmath,amsfonts, amssymb}
\usepackage{graphicx}
\usepackage{xcolor}

\usepackage{thm-restate}
\usepackage{apxproof}
\usepackage{mathtools}
\usepackage{subcaption}
\usepackage{booktabs}
\usepackage{stmaryrd}
\usepackage{algorithm}
\usepackage{algorithmicx}
\usepackage[noend]{algpseudocode}
\usepackage{graphicx}
\usepackage{textcomp}
\usepackage{xcolor}
\usepackage{tikz}
\usetikzlibrary{arrows,automata,positioning,shapes,patterns,arrows.meta,calc}
\usepackage{float}
\usepackage[capitalise]{cleveref}
\usepackage{xspace}
\usepackage{paralist}
\usepackage{todonotes}
\usepackage{ltl}
\usepackage{listings, multicol, multirow}
\usepackage{caption}
\usepackage{wrapfig}
\usepackage{array}
\usepackage{enumitem}

\usepackage{tabularx}

\usepackage{lineno}

\definecolor{pastelorange}{HTML}{F79A3D}
\definecolor{pastelblue}{HTML}{2874ae}
\definecolor{pastelgreen}{HTML}{61B940}
\definecolor{pastelviolet}{HTML}{A370AB}
\definecolor{darkblue}{HTML}{1C4987}
\definecolor{chestnut}{HTML}{A24516}
\definecolor{darkgreen}{HTML}{426A5A}

\newtheorem{requirement}{Requirement}

\newcommand{\PastTSL}{PastTSL\@\xspace}
\newcommand{\pastTSL}{pastTSL\@\xspace}

\newcommand{\pastLTL}{pastLTL\@\xspace}

\newcommand{\stateMachine}{\mathcal{M}}
\newcommand{\winningRegion}{\mathcal{W}}

\newcommand{\nat}{\mathbb{N}}
\newcommand{\true}{\mathit{true}}
\newcommand{\false}{\mathit{false}}
\renewcommand{\phi}{\varphi}
\newcommand{\ldot}{\mathpunct{.}}
\newcommand{\set}[1]{\{#1\}}
\newcommand{\pow}[1]{2^{#1}}
\newcommand{\ap}{\mathit{AP}}

\newcommand{\Next}{\LTLnext}
\newcommand{\Globally}{\LTLsquare}

\newcommand{\Yesterday}{\LTLcircleminus}
\newcommand{\WeakYesterday}{\LTLcircletilde}
\newcommand{\Since}{\LTLsince}
\newcommand{\Historically}{\LTLsquareminus}
\newcommand{\Once}{\LTLdiamondminus}

\newcommand{\cell}[1]{{\normalfont{\texttt{#1}}}}
\newcommand{\update}[2]{\llbracket #1 \leftarrowtail #2 \rrbracket}
\newcommand{\values}{\mathcal{V}}
\newcommand{\functions}{\mathcal{F}}

\newcommand{\predTerm}{\tau_p}
\newcommand{\funcTerm}{\tau_f}
\newcommand{\predTerms}{\mathcal{T}_P}
\newcommand{\funcTerms}{\mathcal{T}_F}
\newcommand{\updateTerms}{\mathcal{T}_U}
\newcommand{\assign}{{\langle \hspace{-1pt}\cdot \hspace{-1pt}\rangle}}
\newcommand{\funcSymbols}{\Sigma_F}
\newcommand{\constSymbols}{\Sigma_F^0}
\newcommand{\predSymbols}{\Sigma_P}
\newcommand{\inputs}{\mathbb{I}}

\newcommand{\cells}{\mathbb{C}}
\newcommand{\cellAssignments}{\mathcal{C}}
\newcommand{\computation}{\varsigma}
\newcommand{\compStream}{\cellAssignments^\omega}
\newcommand{\initial}[1]{\mathit{init}_{#1}}
\newcommand{\inputStream}{\mathcal{I}^\omega}

\newcommand{\eval}{\eta_{\assign}}
\newcommand{\TSLSatisfaction}[1]{\models_{#1}}

\newcommand{\method}[1]{\texttt{\color{darkblue}#1}}
\newcommand{\pmethod}[2]{\method{#1}\texttt{(#2)}}
\newcommand{\field}[1]{\texttt{\color{chestnut}#1}}
\newcommand{\pfield}[2]{\field{#1}\texttt{(#2)}}
\newcommand{\parameter}[1]{\texttt{#1}}

\newcommand{\params}{\mathbb{P}}
\newcommand{\instan}{\mu}

\newcommand{\wini}[1]{\winningRegion_{#1}}
\newcommand{\selfU}{O_\mathit{self}}
\newcommand{\methodProps}{I_\mathit{call}}
\newcommand{\K}{\color{darkgreen!80}K}

\newcommand{\tool}{\textsc{SCSynt}\@\xspace}

\lstdefinelanguage{Solidity}{
	keywords=[1]{anonymous, assembly, assert, balance, break, call, callcode, case, catch, class, constant, continue, constructor, contract, debugger, default, delegatecall, delete, do, else, emit, event, experimental, export, external, false, finally, for, function, gas, if, implements, import, in, indexed, instanceof, interface, internal, is, length, library, log0, log1, log2, log3, log4, memory, modifier, new, payable, pragma, private, protected, public, pure, push, require, return, returns, revert, selfdestruct, send, solidity, storage, struct, suicide, super, switch, then, this, throw, transfer, true, try, typeof, using, value, view, while, with, addmod, ecrecover, keccak256, mulmod, ripemd160, sha256, sha3}, 
	keywordstyle=[1]\color{darkblue}\bfseries,
	keywords=[2]{address, bool, byte, bytes, bytes1, bytes2, bytes3, bytes4, bytes5, bytes6, bytes7, bytes8, bytes9, bytes10, bytes11, bytes12, bytes13, bytes14, bytes15, bytes16, bytes17, bytes18, bytes19, bytes20, bytes21, bytes22, bytes23, bytes24, bytes25, bytes26, bytes27, bytes28, bytes29, bytes30, bytes31, bytes32, enum, int, int8, int16, int24, int32, int40, int48, int56, int64, int72, int80, int88, int96, int104, int112, int120, int128, int136, int144, int152, int160, int168, int176, int184, int192, int200, int208, int216, int224, int232, int240, int248, int256, mapping, string, uint, uint8, uint16, uint24, uint32, uint40, uint48, uint56, uint64, uint72, uint80, uint88, uint96, uint104, uint112, uint120, uint128, uint136, uint144, uint152, uint160, uint168, uint176, uint184, uint192, uint200, uint208, uint216, uint224, uint232, uint240, uint248, uint256, var, void, ether, finney, szabo, wei, days, hours, minutes, seconds, weeks, years},	
	keywordstyle=[2]\color{teal}\bfseries,
	keywords=[3]{block, blockhash, coinbase, difficulty, gaslimit, number, timestamp, msg, data, gas, sender, sig, value, now, tx, gasprice, origin},	
	keywordstyle=[3]\color{pastelviolet!80!black}\bfseries,
	identifierstyle=\color{black},
	sensitive=false,
	comment=[l]{//},
	morecomment=[s]{/*}{*/},
	commentstyle=\color{gray}\ttfamily,
	stringstyle=\color{red!70!black}\ttfamily,
	morestring=[b]',
	morestring=[b]"
}

\newcounter{specification}
\lstnewenvironment{specification}[2]{
	
	\setcounter{lstlisting}{\value{specification}}
	\lstset{caption={\normalsize{#1}},
        label={#2},
        captionpos=b,
        float=t,
        basicstyle = \linespread{0.9}\ttfamily\footnotesize,
	}
}{\addtocounter{specification}{1}}

\newcounter{solidity}
\lstnewenvironment{solidity}[2]{
	
	\setcounter{lstlisting}{\value{solidity}}
	\lstset{caption={\normalsize{#1}},
		label={#2},
        captionpos=b,
        float=t,
        language=Solidity,
        basicstyle = \linespread{0.9}\ttfamily\footnotesize,
	}
}{\addtocounter{solidity}{1}}
\lstnewenvironment{solidityMultiCols}[3]{
	
	\setcounter{lstlisting}{\value{solidity}}
	\lstset{
        language=Solidity,
        multicols=#3,
        basicstyle = \linespread{0.9}\ttfamily\footnotesize,
	}
}{\captionsetup{skip=0pt}\captionof{lstlisting}{\normalsize{#1}}\label{#2}\addtocounter{solidity}{1}}

\begin{document}
\title{Reactive Synthesis of \\ Smart Contract Control Flows
}

\author{}
\institute{}
%
%
\author{Bernd Finkbeiner\inst{1}\orcidID{0000-0002-4280-8441} \and
Jana Hofmann\inst{2}\orcidID{0000-0003-1660-2949}\and
Florian Kohn\inst{1}\orcidID{0000-0001-9672-2398} \and
Noemi Passing\inst{1}\orcidID{0000-0001-7781-043X}}


\institute{CISPA Helmholtz Center for Information Security, Saarbrücken, Germany
\email{\{finkbeiner,florian.kohn,noemi.passing\}@cispa.de}\\
\and
Azure Research, Microsoft, Cambridge, UK\\
\email{t-jhofmann@microsoft.com}}
\maketitle              
\begin{abstract}
Smart contracts are small but highly error-prone programs that implement agreements between multiple parties.
We present a reactive synthesis approach for the automatic construction of smart contract state machines.
Towards this end, we extend temporal stream logic~(TSL) with universally quantified parameters over infinite domains.
Parameterized TSL is a convenient logic to specify the temporal control flow, i.e., the correct order of transactions, as well as the data flow of the contract's fields.
We develop a two-step approach that 1) synthesizes a finite representation of the -- in general -- infinite-state system and 2) splits the system into a compact hierarchical architecture that enables the implementation of the state machine in Solidity.
We implement the approach in our prototype tool \tool, which -- within seconds -- automatically constructs Solidity code that realizes the specified control flow.

\end{abstract}
\subsubsection*{Acknowledgements.}
This work was supported by the European Research Council (ERC) Grant HYPER (No. 101055412) and by DFG grant 389792660 as part of TRR 248.
%
%
%
\section{Introduction}
Smart contracts are small programs that implement digital contracts between multiple parties.
They are deployed on the blockchain and thereby remove the need for a trusted third party that enforces a correct execution of the contract.
Recent history, however, has witnessed numerous bugs in smart contracts, some of which led to substantial monetary losses.
One critical aspect is the implicit state machine of a contract: to justify the removal of a trusted third party -- a major selling point for smart contracts -- all parties must trust that the contract indeed enforces the agreed order of transactions.

Formal methods play a significant role in the efforts to improve the trustworthiness of smart contracts.
Indeed, the \emph{code is law} paradigm is shifting towards a \emph{specification is law} paradigm~\cite{DBLP:journals/corr/abs-2205-07529}.
Formal verification has been successfully applied to prove the correctness of the implicit state machine of smart contracts, for example, by verifying the contract against temporal logic specifications~\cite{DBLP:conf/sp/PermenevDTDV20,DBLP:conf/sp/StephensFMLD21,nehai2018model} or a given state machine~\cite{DBLP:conf/vstte/0001LCPDBNF19}.
Other approaches model the control flow with state machines and construct Solidity code from it~\cite{zupan2020secure,DBLP:conf/fc/MavridouL18,DBLP:journals/insk/CiccioCDGLLMPTW19,DBLP:journals/spe/Lopez-PintadoGD19}. 
Synthesis, i.e., the automatic construction of Solidity code \emph{directly} from a temporal specification, has hardly been studied so far (except for a first step~\cite{DBLP:journals/corr/abs-1906-02906}, see related work).

In this paper, we study the synthesis of smart contracts state machines from temporal stream logic~(TSL), which we equip with universally quantified parameters.
TSL extends linear-time temporal logic~(LTL) with data cells and uninterpreted functions and predicates.
These features enable us to reason about the order of transactions as well as the data flow of the contract's fields.
To distinguish method calls from different callers, we extend the logic with universally quantified parameters.
For example, the following parameterized TSL formula expresses that every voter can only vote once and that a field \field{numVotes} is increased with every vote.
\begin{lstlisting}[numbers=none]
	(*$\forall$*)m. (*$\Globally$*)((*$\pmethod{vote}{m}$*) -> [[(*$\field{numVotes}$*) <- (*$\field{numVotes}$*)$\,$+$\,$1]] && (*$\Next \Globally\, $*)!$\,$(*$\pmethod{vote}{m}$*)) 
\end{lstlisting}
The above formula demonstrates the challenges associated with parameterized TSL synthesis.
First of all, a part of the formula restricts the allowed method calls, which are inputs in the synthesis problem.
To make specifications realizable, we restrict ourselves to safety properties, which we express in the past-time fragment of parameterized TSL.
Second, as the contract might interact with arbitrarily many voters, the above formula ranges over an infinite domain.
However, we need to find a finite representation of the system that can be translated into feasible Solidity code.

We tackle this challenge in two steps.
First, we translate the parameterized pastTSL formula to pastTSL to synthesize a finite representation of the system.
Unfortunately, we show that the realizability problem of pastTSL is undecidable, even without parameters.
As a remedy, we employ a sound approximation in LTL~\cite{FinkbeinerKPS19} to make synthesis possible.

In a second step, we split the resulting state machine into a hierarchical structure of smaller, distributed state machines.
This architecture can be interpreted as an infinite-state system realizing the original formula.
It also minimizes the number of transactions needed to keep the system up to date at runtime.

We implement the approach in our prototype \tool, which, due to the past-time fragment, leverages efficient symbolic algorithms.
We specify ten different smart contract specifications and obtain an average synthesis time of two seconds.
Our largest specification is based on Avolab's NFT auction~\cite{NFTReference} and produces a state machine with 12 states in 12 seconds.
To summarize, we
\begin{itemize}
	[leftmargin=15pt]
	\item show how to specify smart contract control flows in parameterized \pastTSL,
	\item prove undecidability of the general realizability problem of \pastTSL,
	\item and present a sound (but necessarily incomplete) synthesis approach for parameterized pastTSL formulas that generates a hierarchy of state machines to enable a compact representation of the system in Solidity.
\end{itemize}

\paragraph{Related Work.}
Formal approaches for smart contracts range from the automatic construction of contracts from state machines~\cite{DBLP:conf/fc/MavridouL18,MavridouLSD19}, over the verification against temporal logics~\cite{DBLP:conf/sp/PermenevDTDV20,DBLP:conf/sp/StephensFMLD21,nehai2018model} and state machines~\cite{DBLP:conf/vstte/0001LCPDBNF19,DBLP:journals/corr/abs-1812-08829}, to deductive verification approaches~\cite{dharanikota2021celestial,DBLP:conf/ndss/KalraGDS18}.
Closest to our work is a synthesis approach based on LTL specifications~\cite{DBLP:journals/corr/abs-1906-02906}.
The approach does not reason about the contract's data: neither about the current value of the fields, nor about parameters like the method's caller.
To quote the authors of~\cite{DBLP:journals/corr/abs-1906-02906}: the main challenge in the synthesis of smart contracts is ``how to strike a balance between simplicity and expressivity [...] to allow effective synthesis of practical smart contracts''.
In this paper, we opt for a more expressive temporal logic and simultaneously aim to keep the specifications readable.

TSL has been successfully applied to synthesize FPGA controllers~\cite{DBLP:journals/corr/abs-2101-07232} and functional reactive programs~\cite{DBLP:conf/haskell/Finkbeiner0PS19}.
To include domain-specific reasoning, TSL has been extended with theories~\cite{FinkbeinerHP22} and SMT solvers~\cite{RoderickTSLTheories}.
A recent approach combines TSL reactive synthesis with SyGus to synthesize implementations for TSL's uninterpreted functions~\cite{cispa3674}.
Parameterized synthesis has so far focused on distributed architectures parameterized in the number of components~\cite{JacobsB14,Party,KhalimovJB13,MarkgrafHLNN20}.
Orthogonal to this work, these approaches rely on a reduction to bounded isomorphic synthesis~\cite{JacobsB14,Party,KhalimovJB13} or apply a learning-based approach~\cite{MarkgrafHLNN20}.

\paragraph{Overview.}
We first provide some brief preliminaries on state machines, reactive synthesis, and TSL.
In \Cref{sec:paramTSL}, we introduce parameterized TSL and demonstrate how it can be used for specifying smart contract control flows.
Subsequently, we discuss the high-level idea and associated challenges of our synthesis approach in \Cref{sec:approach} and discuss synthesis from plain pastTSL in \Cref{sec:pastTSL_synthesis}.
demonstrate how to specify smart contracts using pure pastTSL and prove the undecidability of its realizability problem.
We proceed with the main part of the approach, a splitting algorithm for state machines, in \Cref{sec:paramSynt}.
Finally, we discuss the implementation of \tool and its evaluation in \Cref{sec:impl_eval}.

\section{Preliminaries}
\label{sec:prelims}
\label{sec:background}

We assume familiarity with linear-time temporal logic (LTL). A definition with past-time temporal operators can be found in~\cite{recursiveEvaluation,pastLTL}.
We only assume basic knowledge about smart contracts; for an introduction we refer to~\cite{Ethereum}.

\subsection{State Machines, Safety Properties and Reactive Synthesis}
\label{subsec:prelim_syn_sm}
We give a brief introduction to Mealy machines, safety properties, and reactive synthesis.
In this work, we represent smart contract control flows as \emph{Mealy state machines}~\cite{Mealy55}, which separate the alphabet into inputs $I$ and outputs $O$.
A Mealy machine $\stateMachine$ is a tuple $(S, s_0, \delta)$ of states $S$, initial state $s_0$, and transition relation $\delta \subseteq S \times I \cup O \times S$.
For a compact representation, we attach the outputs also to transitions, not to the states.
We call $\stateMachine$ \emph{finite-state} if both $\Sigma = I \cup O$ and $S$ are finite, and \emph{infinite-state} otherwise.
An infinite sequence $t \in \Sigma^\omega$ is a \emph{trace} of $\stateMachine$ if there is an infinite sequence of states $r \in S^\omega$ such that $r[0] = s_0$ and $(r[i],t[i],r[i+1]) \in \delta$ for all points in time $i \in \mathbb{N}$.
A finite sequence of states $r \in S^+$ results in a finite trace $t\in\Sigma^+$.

In this paper, we work with specifications that are \emph{safety properties}(see, e.g., \cite{KupfermanV99,DBLP:journals/tse/Lamport77}). A safety property can be equivalently expressed as a Mealy machine $\stateMachine$ that describes the set of traces that satisfy the property ($\stateMachine$ is called the \emph{safety region} of the property).
For a safety specification, the \emph{reactive synthesis} problem is to determine the \emph{winning region}, i.e., the maximal subset of its safety region such that for every combination of state and input, there is a transition into said subset.
A \emph{strategy} is a subset of the winning region such that in each state, there is exactly one outgoing transition for every input.

\subsection{Past-time Temporal Stream Logic}
PastTSL is the past-time variant of TSL~\cite{FinkbeinerKPS19}, a logic that extends LTL with cells that can hold data from a possibly infinite domain.
To abstract from concrete data, TSL includes uninterpreted functions and predicates.
\emph{Function terms~$\funcTerm \in \funcTerms$} are recursively defined by
\[
\funcTerm \Coloneqq \texttt{s}~ \mid ~f~ \tau^1_f \ldots ~\tau^n_f
\]
where \texttt{s} is either a cell $\cell{c} \in \cells$ or an input $i \in \inputs$, and $f \in \funcSymbols$ is a function symbol.
\emph{Constants} $\constSymbols \subseteq \funcSymbols$ are 0-ary function symbols.
\emph{Predicate terms~$\predTerm \in \predTerms$} are obtained by applying a predicate symbol $p \in \predSymbols$ with $\predSymbols \subseteq \funcSymbols$ to a tuple of function terms.
\PastTSL formulas are built according to the following grammar:
\[ 
\varphi, \psi \Coloneqq \neg \varphi \,\mid\, \varphi \land \psi \,\mid\, \Yesterday \varphi \,\mid\, \varphi \Since \psi \,\mid\, \predTerm \,\mid\, \update{\cell{c}}{\funcTerm}
\]
An \emph{update term} $\update{\cell{c}}{\tau_f} \in \updateTerms$ denotes that cell $\cell{c}$ is overwritten with $\tau_f$.
The temporal operators are called ``Yesterday'' $\Yesterday$ and ``Since''~$\Since$.
Inputs, function symbols, and predicate symbols a purely syntactic objects. To assign meaning to them, let $\values$ be the set of values with $\mathbb{B} \subseteq \values$.
We denote by $\mathcal{I} : \inputs \to \values$ the evaluation of inputs.
An \emph{assignment function} $\assign: \funcSymbols \rightarrow \functions$
assigns function symbols to functions $\functions = \bigcup_{n \in \nat} \values^n \rightarrow \values$.

The type $\cellAssignments = \cells \to \funcTerms$ describes an update of all cells.
For every cell $\cell{c} \in \cells$, let $\initial{\cell{c}}$ be its initial value.
The evaluation function $\eval: \compStream \times \inputStream \times \mathbb{N} \times \funcTerms \rightarrow \values$ evaluates a function term at point in time $i$ with respect to an \emph{input stream} $\iota \in \inputStream$ and a \emph{computation} $\computation \in \compStream$:
\begin{align*}
	&\eval(\computation,\iota,i,\cell{s}) \coloneqq \begin{cases}
		\iota~i~\cell{s} & \text{if } \cell{s} \in \inputs \\
		\initial{\cell{s}} & \text{if } \cell{s} \in \cells \land i = 0 \\
		\eval(\computation,\iota,i-1,\computation~(i-1)~\cell{s}) & \text{if } \cell{s} \in \cells \land i>0
	\end{cases} \\
	& \eval(\computation,\iota,i,f~ \tau_0 \dots \tau_{m-1}) \coloneqq \langle f \rangle ~ \eval(\computation,\iota,i, \tau_0) ~\dots ~\eval(\computation,\iota,i,\tau_{m-1})
\end{align*}
Note that $\iota~i~\cell{s}$ denotes the value of $\cell{s}$ at position $i$ according to $\iota$.
Likewise, $\computation ~ i~ \cell{s}$ is the function term that $\computation$ assigns to $\cell{s}$ at position $i$.
With the exception of update and predicate terms, the semantics of pastTSL is similar to that of LTL.
\allowdisplaybreaks
\begin{alignat*}{3}
		&\computation, \iota, t \TSLSatisfaction{\assign} \neg \varphi ~ &\quad \text{iff} \quad & ~ \computation, \iota, t \not\TSLSatisfaction{\assign} \varphi\\
		&\computation, \iota, t \TSLSatisfaction{\assign} \varphi \land \psi ~ &\quad \text{iff} \quad & ~ \computation, \iota, t \TSLSatisfaction{\assign} \varphi \text{ and } \computation, \iota, t \TSLSatisfaction{\assign} \psi\\
		&\computation, \iota, t \TSLSatisfaction{\assign} \Yesterday \varphi ~ &\quad \text{iff} \quad & ~ t > 0 \land \computation, \iota,  t-1 \TSLSatisfaction{\assign} \varphi\\
		&\computation, \iota, t \TSLSatisfaction{\assign} \varphi \Since \psi ~ &\quad \text{iff} \quad & ~ \exists ~ 0 \leq t' \leq t. ~ \computation, \iota, t' \TSLSatisfaction{\assign} \psi \text{ and }\\
		& & & \quad \forall t' < k \leq t. ~ \computation, \iota, k \TSLSatisfaction{\assign} \varphi \\
	&\computation, \iota, t \TSLSatisfaction{\assign} \update{\cell{v}}{\tau} ~ &\quad \text{iff} \quad & ~ \computation~t~\cell{v} \equiv \tau\\
	&\computation, \iota, t \TSLSatisfaction{\assign} p~ \tau_0 \dots \tau_{m} ~ &\quad \text{iff} \quad & ~ \eval(\computation,\iota,t,p~ \tau_0 \dots \tau_{m-1})
\end{alignat*}
We use $\equiv$ to syntactically compare two terms.
We derive three additional operators: $\WeakYesterday \varphi := \neg \Yesterday \neg \varphi$, $\Once \phi := \true \Since \phi$, and $\Historically \phi := \lnot \Once \lnot \phi$. 
The difference between $\Yesterday$ and ``Weak Yesterday'' $\WeakYesterday$ is that $\Yesterday$ evaluates to $\false$ in the first step and $\WeakYesterday$ to $\true$.
We use pastTSL formulas to describe safety properties.
Therefore, we define that 
computation $\computation$ and an input stream $\iota$ satisfy a pastTSL formula $\varphi$, written $\computation, \iota \TSLSatisfaction{\assign} \varphi$,
if $\forall i \in \nat \ldot \computation, \iota, i \TSLSatisfaction{\assign} \psi$.

The realizability problem of a \pastTSL formula $\psi$ asks whether there exists a strategy that reacts to predicate evaluations with cell updates according to $\psi$. Formally, a strategy is a function $\sigma : (2^{\predTerms})^+ \rightarrow \cellAssignments$.
For $\iota \in \inputStream$, we write $\sigma(\iota)$ for the computation obtained from $\sigma$:
$$\sigma(\iota)(i) = \sigma(\{\predTerm \in \predTerms \mid \eval(\sigma(\iota), \iota, 0, \predTerm)\} \ldots \{\predTerm \in \predTerms \mid \eval(\sigma(\iota), \iota, i, \predTerm)\})$$
Note that in order to define $\sigma(\iota)(i)$, the definition uses $\sigma(\iota)$. This is well-defined since the evaluation function $\eval(\computation, \iota, i, \tau)$ only uses $\computation~0 \ldots \computation\,(i-1)$.

\begin{definition}[\cite{FinkbeinerKPS19}]
	A \pastTSL formula $\psi$ is \emph{realizable} if, and only if, there exists a strategy $\sigma : (2^{\predTerms})^+ \rightarrow \cellAssignments$ such that for every input stream $\iota \in \inputStream$ and every assignment function $\assign: \funcSymbols \rightarrow \functions$ it holds that $\sigma(\iota), \iota \TSLSatisfaction{\assign} \psi$.
\end{definition}

\section{Parameterized TSL for Smart Contract Specifications}
\label{sec:paramTSL}  
In this section, we introduce parameterized pastTSL and show how the past-time fragment of the logic can be used for specifying smart contract state machines.

\subsection{Parameterized TSL}
\label{subsec:paramTSL}
Parameterized TSL extends TSL with universally quantified parameters.
Let $P$ be a set of parameters and $\cells_P$ a set of parameterized cells, where each cell is of the form $\cell{c}(p_1, \ldots, p_m)$ with $p_1, \ldots, p_m \in P$.
A parameterized TSL formula is a formula $\forall p_1, \ldots, \forall p_n \ldot \psi$, where $\psi$ is a TSL formula with cells from $\cells_P$ and which may use parameters as base terms in function and predicate terms.
We require that the formula is closed, i.e., every parameter occurring in $\psi$ is bound in the quantifier prefix.

Parameterized TSL formulas are evaluated with respect to a domain $\params$ for the parameters. 
We use a function $\instan: P \rightarrow \params$ to instantiate parameters.
Given a parameterized TSL formula $\forall p_1, \ldots, \forall p_n \ldot \psi$, $\psi[\instan]$ is the formula obtained by replacing all parameters according to $\instan$.
To simplify our constructions, we want $\psi[\instan]$ to be a TSL formula. Therefore, we assume that $\params$ is a subset of the set of constants and that $\cell{c}(\instan(p_1), \ldots, \instan(p_m)) \in \cells$, i.e., the instantiation of a parameterized cell refers to a normal, non-parameterized cell.
Given a computation $\computation$ and an input stream $\iota$, we define $\computation, \iota \models \forall p_1, \ldots, \forall p_n \ldot \psi$ iff $\forall \instan: P \to \params \ldot \computation, \iota \models \psi[\instan]$.

\subsection{Example: ERC20 Contract}
We illustrate how parameterized pastTSL can be used to specify the state machine logic of smart contract with an ERC20 token system.
An ERC20 token system provides a platform to transfer tokens between different accounts.
We follow the Open Zeppelin documentation~\cite{ERC20Spec}.
The special feature of the contract is the possibility to transfer not only tokens from one's own account but, after approval, also from a different account.
The core contract consists of methods \method{transfer}, \method{transferFrom}, and \method{approve}.
We do not model getters like \method{totalSupply} or \method{balanceOf} as they are not relevant for the temporal behavior of the contract.
The Open Zeppelin ERC20 contract describes various extensions to the core contract, one of which is the ability to pause transfers.
We distinguish between pausing transfers globally (\method{pause}) and from one's own account (\pmethod{pause}{m}).

Our specifications describe the temporal control flow of the contract's method calls and the data flow of its fields.
We distinguish between \emph{requirements}, \emph{obligations}, and \emph{assumptions}.
Requirements enforce the right order of method calls with correct arguments.
Obligations describe the data flow in the fields of the contract.
Assumptions restrict the space of possible predicate evaluations.
For this example, we do not need any assumptions. A typical assumption in other specifications would be that $\field{x} > \field{y}$ and $\field{y} > \field{x}$ cannot hold at the same time.

To emphasize that all past-time formulas are required to hold globally, we add a $\Globally$ operator to formulas.
We use two parameters \parameter{m} and \parameter{n}, where \parameter{m} always refers to the address from which tokens are subtracted and parameter \parameter{n}, whenever different from \parameter{m}, to the address that initiates the transfer.
We start with the requirements. First, any transfer from \parameter{m} must be backed by sufficient funds.
\begin{lstlisting}[numbers=none]
	(*$\Globally$*)((*\pmethod{transfer}{m}*) || (*\pmethod{transferFrom}{m,n}*) -> suffFunds(m, arg@amount))
\end{lstlisting}
Second, no method call can happen after \method{pause} until \method{unpause} is called:
\begin{lstlisting}[numbers=none]
	(*$\Globally$*)((*\pmethod{transferFrom}{m,n}*) || (*\pmethod{transfer}{m}*) || (*\pmethod{approve}{m,n}*) || (*\pmethod{localPause}{m}*) || (*\pmethod{localUnpause}{m}*) -> (!(*\method{pause}*) (*$\Since$*) (*\method{unpause}*)) || (*$\Historically$*) !(*\method{pause}*))
\end{lstlisting}
In contrast, \pmethod{localPause}{m} only stops method calls from \parameter{m}'s account:
\begin{lstlisting}[numbers=none]
	(*$\Globally$*)((*\pmethod{transferFrom}{m,n}*) || (*\pmethod{transfer}{m}*) || (*\pmethod{approve}{m,n}*)(*\break*) -> ((!(*\pmethod{localPause}{m}*)) (*$\Since$*) (*\pmethod{localUnpause}{m}*)) || (*$\Historically$*) !(*\pmethod{localPause}{m}*))
\end{lstlisting}
Finally, \method{pause} and \method{unpause} can only be called by the owner of the contract. Additionally, they cannot be called twice without the respective other in between and \method{unpause} cannot be called if \method{pause} has not been called at least once.
\begin{lstlisting}[numbers=none]
	(*$\Globally$*)((*$\method{unpause}$*) -> msg.sender$\,$=$\,$owner() &&$\,$ (*$\Yesterday$*)(!$\,$(*$\method{unpause}$*) (*$\Since$*) (*$\method{pause}$*)))
	(*$\Globally$*)((*$\method{pause}$*) -> msg.sender$\,$=$\,$owner() &&$\,$ (*$\Yesterday$*)(!$\,$(*$\method{pause}$*)$\,$(*$\Since$*)$\,$(*$\method{unpause}$*))$\,$||$\,$(*$\WeakYesterday \, \Historically$*)$\,$!$\,$(*$\method{pause}$*))
\end{lstlisting}
\texttt{mgs.sender} is an input, whereas \texttt{owner()} is a constant.
For the obligations, we need to make sure that the \field{approved} field is updated correctly. We use TSL's cell mechanism to model fields and use parameterized cells for mappings.
\begin{lstlisting}[numbers=none]
	(*$\Globally$*)((*\pmethod{approve}{m,n}*)$\,$->$\,\update{\pfield{approved}{m,n}}{\text{arg@amount}}$)
	(*$\Globally$*)((*\pmethod{transferFrom}{m,n}*)$\,$->$\,\update{\pfield{approved}{m,n}}{\pfield{approved}{m,n}\text{-arg@amount}}$)
	(*$\Globally$*)(!((*\pmethod{transferFrom}{m}*)$\,$||$\,$(*\pmethod{approve}{m,n}*))$\,$->$\,\update{\pfield{approved}{m,n}}{\pfield{approved}{m,n}}$)
\end{lstlisting}
Transitions that do not change the content of a cell are indicated by self-updates like $\update{\pfield{approved}{m,n}}{\pfield{approved}{m,n}}$.

\section{Synthesis Approach}
\label{sec:approach}
The synthesis goal of this paper is to construct a state machine that satisfies parameterized pastTSL specifications like the one given in the last section.

\subsection{Problem Statement}
\label{subsec:problem}
Our specifications are split into assumptions $\phi_A$, requirements $\phi_R$, and obligations $\phi_O$, all of which are parameterized pastTSL formulas.
Each of them can be given as invariant $\phi^\mathit{inv}$ or as initial formula $\phi^\mathit{init}$.
For synthesis, we compose them to the following formula, which, according to the definition of (parameterized) pastTSL, is required to hold globally.
\begin{align*}
	\phi \coloneqq \, &\forall \parameter{p}_1, \ldots, \parameter{p}_m \ldot \\
	& \quad (\WeakYesterday \false \rightarrow \phi^\mathit{init}_A \land \phi^\mathit{init}_R) \land (\Historically (\phi^\mathit{inv}_A \land \phi^\mathit{inv}_R)) \rightarrow (\WeakYesterday \false \rightarrow \phi^\mathit{init}_O) \land \phi^\mathit{inv}_O
\end{align*}
\begingroup
\setlength{\intextsep}{4pt}
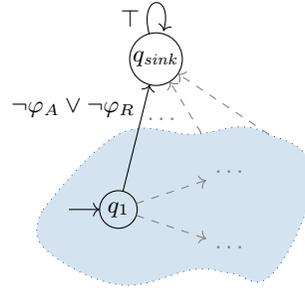
\begin{wrapfigure}{r}{0pt}
	\footnotesize
	\centering
	\begin{tikzpicture}[initial text=, 
		->,
		node distance=1.5cm,
		state/.style = {circle, draw, minimum size=5mm,
			inner sep=0pt, outer sep=0pt},
		state/.default = 6pt  
		] 
		\draw [-,dotted,pastelblue,fill=pastelblue!20]  plot[smooth, tension=.7] coordinates {(-1,0) (0, 1) (1,1) (1.6, 1.1) (2, 1) (2.5, 0.7) (2.5, 0) (2, -0.9) (1.5,-0.85) (1,-0.8) (0,-1) (-1, -0.5) (-1,0)};	
		
		\node[state, initial left] at (0,0) (q1) {$q_1$};		
		\node[color=gray] at (1.5,0.5) (q3) {$\dots$};			
		\node[color=gray] at (1.5,-0.5) (q4) {$\dots$};					
		\node[color=gray] at (0.6,1.2) {$\dots$};		
		\node[state, minimum size=7mm,] at (0.5,2) (q2) {$q_{\mathit{sink}}$};
		
		\draw[color=gray, dashed]  (q1) edge (q3) (q1) edge (q4);
		\draw (q2) edge[loop above] node[left, xshift=-1mm, yshift=-2mm]{$\top$} (q2);	
		\draw (q1) edge node[left,xshift=1mm,yshift=4mm]{$\neg \phi_A \lor \neg \phi_R$} (q2);
		\draw[color=gray, dashed] (2, 1) to (q2);
		\draw[color=gray, dashed] (1.1, 1.02) to (q2);
		
	\end{tikzpicture}
	\caption{Sketch of the system synthesized from $\phi$. The dotted blue area implements the contract.}
	\label{fig:synthesis_result}
\end{wrapfigure}%
Here, $\parameter{p}_1, \ldots, \parameter{p}_m$ are the parameters occurring in the inner formulas $\phi^\mathit{init}_A$, $\phi^\mathit{init}_R$, $\phi^\mathit{init}_O$, $\phi^\mathit{inv}_A$, $\phi^\mathit{inv}_R$, and $\phi^\mathit{inv}_O$.
We use $\WeakYesterday \false$ to refer to the first position of a trace.

It might seem counter-intuitive that we include requirements on the left side of the implication.
The reason is that requirements describe a monitor on the method calls, which, from a synthesis perspective, constitute system inputs. Thus, if we conjuncted requirements with obligations, the specification would be unrealizable.
Instead, we leverage the fact that all specifications describe safety properties.
Thus,
state machines satisfying $\phi$ have a shape as depicted in \Cref{fig:synthesis_result}.
Whenever an assumption or a requirement is violated, the machine enters an accepting sink state.
To obtain the desired result, we reject any method call for which the system moves to the sink state. Like this, the remaining system enforces the requirements on method calls and also satisfies the obligations.
For the rest of the paper, we depict state machines synthesized from $\phi$ without the sink state.

\paragraph{On Safety Properties.}
We restrict ourselves to safety properties for three reasons.
First of all, as we consider the synthesis problem, our requirements can only describe a monitor on the method calls. Liveness properties are known not to be monitorable.
For future work, one could consider model-checking the synthesized state machine with regard to liveness properties like ``eventually, method X is \emph{callable}".
Second, the restriction to safety automata enables the splitting algorithm described in \Cref{sec:paramSynt}, which is essential for our approach in order to efficiently implement the state machine in Solidity.
Lastly, synthesis from safety properties is less complex than full LTL synthesis (c.f. \Cref{sec:impl_eval}), which enables us to synthesize non-trivial state machines within seconds.

\subsection{High-Level Description of the Approach}
\label{subsec:high-level-algorithm}
\paragraph{Challenges.}
We need to address two major challenges. First, as parameters range over an infinite domain $\params$, parameterized pastTSL formulas describe (in general) \emph{infinite-state systems}.
Second, even if we managed to synthesize some representation of the infinite-state system, we still need to translate it to  Solidity code.
In Solidity, every computation costs gas. Therefore, we need to find a compact representation of the system that minimizes the number of computation steps needed to update the system after a method call.

\paragraph{Approach in a Nutshell.}
\begin{figure}[t]
	\small
	\centering
	\begin{tikzpicture}[initial text=, 
		->,
		node distance=3.6cm,
		state/.style = {circle, draw, minimum size=6mm,
			inner sep=0pt, outer sep=0pt, font=\normalsize}
		] 
		
		\node[state, initial below] (s1) {$s_1$};
		\node[state, right of = s1] (s2) {$s_2$};
		\node[state, left of = s1] (s3) {$s_3$};
		\node[state, right of = s2] (s4) {$s_4$};
		
		\draw  (s1) edge[bend left=15] node[above]{\pmethod{localPause}{m}} (s2)
		(s1) edge[bend right=15,align=center] node[above]{\method{pause} $\land$\\ \texttt{sender = owner()}} (s3)
		(s2) edge[bend left=15,align=center] node[below]{\pmethod{localUnpause}{m}} (s1)
		(s3) edge[bend right=15,align=center] node[below]{\method{unpause} $\land$\\ \texttt{sender = owner()}} (s1)
		(s2) edge[bend left=15,align=center] node[above]{\method{pause} $\land$\\ \texttt{sender = owner()}} (s4)
		(s4) edge[bend left=15,align=center] node[below]{\method{unpause} $\land$\\  \texttt{sender = owner()}} (s2)
		(s1) edge[loop above, looseness =20,align=left] node[xshift=13mm]{\pmethod{approve}{m,n} $\lor$ (\pmethod{transfer}{m} $\land$ \texttt{suffFunds(m,a)}) $\lor$ \\ (\pmethod{transferFrom}{m,n} $\land$ \texttt{suffFunds(m,a)} $\land$ \pfield{approved}{m,n} $\geq$ \texttt{a})} (s1);
	\end{tikzpicture}
	\caption{System $\winningRegion$ for the ERC20 contract. Irrelevant predicates and all cell updates are omitted for readability. We also write \texttt{a} instead of \texttt{arg@amount}.}
	\label{fig:w_erc20}
\end{figure}
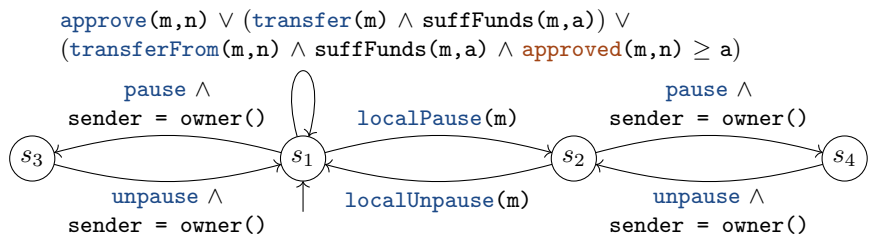
We address these challenges in two steps.
First, we interpret the specification as being unquantified, i.e., we remove all quantifiers and tread the parameters as normal constants (e.g., in case of \texttt{suffFunds(m, arg@amount)}) or as part of the cells' name (e.g., in case of \pfield{approved}{m,n}).
Like that, we obtain a plain pastTSL formula that describes the finite-state system representing the correct control flow for every parameter instantiation.
We synthesize the winning region from that formula, which we call $\winningRegion$.
For the running ERC20 example, $\winningRegion$ can be found in \Cref{fig:w_erc20}.

Of course, the contract can be in different states of $\winningRegion$ depending on the parameter instantiation.
In theory, we would therefore like to keep the necessary number of copies of $\winningRegion$.
For example, if \pmethod{approve}{m=1,n=2} is called, we would execute the corresponding transition in system $\winningRegion_{(\parameter{m} = 1, \parameter{n} = 2)}$.
The problem with this naive approach is that calling a method parameterized with only a subset of the parameters would lead to updates of several systems.
For example, 
if \pmethod{localPause}{m=1} is called, this would have to be recorded in all $\winningRegion_{(\parameter{m} = 1, \parameter{n} = v)}$ for any value $v$ of \parameter{n} observed so far.
Updating all these state machines after each method call would lead to a quick explosion of the gas consumption in Solidity.
Instead, addressing the second challenge, we split $\winningRegion$ into a hierarchical structure of state machines, one for each subset of parameters. As a result, we only have to update a single state machine per method call and still maintain the correct state of each instance (we describe this approach in more detail in \Cref{subsec:idea}).
To summarize, we proceed as follows.
\begin{enumerate}
	\item Interpret the parameterized pastTSL formula $\phi$ as a pastTSL formula $\psi$ and synthesize the winning region $\winningRegion$ from it.
	\item Split $\winningRegion$ into a hierarchical structure $\winningRegion_1, \ldots, \winningRegion_n$ and show how these systems can be interpreted as an infinite-state machine $\stateMachine$ satisfying $\phi$.
	\item Generate Solidity code that implements transitions according to $\winningRegion_1, \ldots, \winningRegion_n$.
\end{enumerate}

In the following sections, we discuss each of these steps in detail.

\section{PastTSL Synthesis}
\label{sec:pastTSL_synthesis}
Let $\phi$ be a parameterized pastTSL formula as described in \Cref{subsec:problem}.
We first translate $\phi$ to pastTSL. This is easy: just remove all quantifiers and interpret parameters as constants (i.e., $P \subseteq \constSymbols$) and parameterized cells as normal cells (i.e., $\cells_P \subseteq \cells$).

Unfortunately, even though past-time fragments usually simplify logical problems, we establish that the realizability problem of pastTSL is undecidable.
We obtain this result by a reduction from the universal halting problem of lossy counter machines~\cite{DBLP:journals/tcs/Mayr03}.

An \emph{n-counter machine ($n$CM)} consists of a finite set of instructions $l_1, \ldots, l_m$, which modify $n$ counters $c_1, \ldots, c_n$. Each instruction $l_i$ is of one of the following forms, where $1 \leq x \leq n$ and  $1 \leq j,k \leq m$.
\begin{itemize}[leftmargin=20pt]
	\item \lstinline|$l_i$: $c_x \coloneqq c_x+1$; goto $l_j$|
	\item \lstinline|$l_i$: if $c_x = 0$ then goto $l_j$ else $c_x \coloneqq c_x-1$; goto $l_k$|
	\item \lstinline|$l_i$: halt|
\end{itemize}
A configuration of a $n$CM is a tuple $(l_i, v_1, \ldots, v_n)$, where $l_i$ is the next instruction to be executed, and $v_1, \ldots, v_n$ denote the values of the counters.
Compared to non-lossy $n$CMs, the counters of a lossy $n$CM may spontaneously decrease.
We employ a version of lossiness where a counter can become zero if it is tested for zero (see \cite{DBLP:journals/tcs/Mayr03} for details).
A lossy $n$CM halts from an initial configuration if it eventually reaches a state with the halting instruction.

\begin{theorem}
	The pastTSL realizability problem is undecidable.
\end{theorem}

\begin{proof}
	We reduce from the universal halting problem of lossy $n$CMs, which is undecidable~\cite{DBLP:journals/tcs/Mayr03}.
	We spell out the main ideas. 
	Our formulas consist of one constant $z()$, one function $f$, and one predicate $p$. There are no inputs.
	Applying an idea from~\cite{DBLP:conf/time/LisitsaP05}, we use two cells for every counter $c_x$: $c_x^\mathit{inc}$ to count increments and~$c_x^\mathit{dec}$ to count decrements. Applying~$f$ to~$c_x^\mathit{inc}$ increments the counter, applying~$f$ to~$c_x^\mathit{dec}$ decrements it. If the number of increments and decrements is equal, the counter is zero. In TSL, we use the formula $\psi^0_x \coloneqq p(c_x^\mathit{inc}) \leftrightarrow p(c_x^\mathit{dec})$ to test if a counter is zero. Note that if the counter really is zero, then the test for zero \emph{must} evaluate to $\true$ by the TSL semantics. For all other cases, it \emph{may} evaluate to $\true$. If the equivalence evaluates to $\true$ even though the counter is non-zero, we interpret it as a spontaneous reset.
	Initially, the value of the counters need to be arbitrary.
	We reflect this by making no assumptions on the first step, thereby allowing the strategy to set the counter cells to any valid function term $f^*(z())$.	
	We use $n$ cells $l_1, \ldots , l_n$ for encoding the instructions. Globally, all instruction cells but the one indicating the next instruction, indicated by $\update{l_i}{f(l_i)}$, need to self-update.
	We spell out the encoding of an instruction of the second type.
	\begin{align*}
		\Globally (\Yesterday \update{l_i}{f(l_i)} \rightarrow &(\psi^0_x \rightarrow \update{l_j}{f(l_j)} \land \update{c_x^\mathit{inc}}{z()} \land \update{c_x^\mathit{dec}}{z()})\\
		& \land (\neg \psi^0_x \rightarrow \update{l_k}{f(l_k)} \land \update{c_x^\mathit{inc}}{c_x^\mathit{inc}} \land \update{c_x^\mathit{dec}}{f(c_x^\mathit{dec})})\!)
	\end{align*}
	The formula tests if the instruction to be executed is $l_i$. If so, we test the counter $c_x$ for zero and set the corresponding cell to $z()$ if that is the case. Furthermore, the correct next instruction is updated by applying $f$. Finally, we encode that we never reach a halting state: $\Globally \neg \update{l_\mathit{halt}}{f(l_\mathit{halt})}$.
	The resulting \pastTSL formula is realizable if, and only if, there is an initial state such that the machine never halts. Thus, undecidability of the \pastTSL realizability problem follows.
\end{proof}

\subsection{PastTSL Synthesis via PastLTL Approximation}
\label{par:tsl_approx}
As pastTSL realizability is undecidable, we have to approximate the synthesis problem. To do so, we employ a reduction proposed in~\cite{FinkbeinerKPS19}, which approximates TSL synthesis in LTL, for which realizability is decidable.
The reduction replaces all predicate terms and update terms of a TSL formula $\psi$ with unique atomic propositions, e.g., $a_{p\_x}$ for $p(\cell{x})$ and $a_{x\_\mathit{to}\_f\_x}$ for $\update{\cell{x}}{f(\cell{x})}$.
Additionally, the reduction adds a formula that ensures that every cell is updated with exactly one function term in each step.
Given a pastTSL formula $\psi$, the reduction produces an LTL approximation $\psi_\text{LTL}$ that also falls into the past-time fragment.
The reduction is sound but not complete~\cite{FinkbeinerKPS19}, i.e., $\psi$ might be realizable even if $\psi_\text{LTL}$ is not.
For the smart contract specifications we produced for our evaluation, however, we never encountered spurious unrealizability.

Let $\ap$ be the set of atomic propositions of $\psi_\text{LTL}$. 
From every trace $t$ over $\ap$, we can directly generate a computation $\mathit{comp}(t) \in \compStream$ as follows: 
$$
	\mathit{comp}(t)(i)(\cell{c}) = \funcTerm \quad \text{if } a_{\cell{c}\_\mathit{to}\_{\tau_f}} \in t(i)
$$
For the other direction, given a computation $\computation$, an input stream $\iota$, and an assignment function $\assign$, we write $\mathit{LTL}(\iota, \computation, \assign)$ for the corresponding trace over $\ap$.
$$
	\mathit{LTL}(\iota, \computation, \assign) (i) = \set{\{a_{\tau_p} \mid \eval(\computation, \iota, i, \predTerm)\}} \cup \set{a_{\cell{c}\_\mathit{to}\_{\tau_f}} \mid \computation(i)(\cell{c}) = \funcTerm}
$$
The following proposition follows from the soundness of the approximation.
\begin{proposition}
\label{prop:soundness}
	For every assignment function $\assign$, input stream $\iota$, and computation $\computation$, $\mathit{LTL}(\iota, \computation, \assign) \models \psi_\text{LTL}$ iff $\computation, \iota \models_\assign \psi$.
\end{proposition}

\paragraph{Parameterized Atomic Propositions.}
In our case, the pastTSL formula $\psi$ is obtained from a parameterized pastTSL formula $\phi$.
Thus, the atomic propositions of $\psi_\text{LTL}$ contain parameters, e.g., $a_{\mathit{transferFrom}\_m\_n}$.
To enable correctness reasoning in the next section, we lift the instantiation of parameters to the level of atomic propositions and LTL formulas.

For $a \in \ap$, we write $a(p_1, \ldots, p_m)$ if $a$ contains parameters $p_1, \ldots, p_m$.
We usually denote the sequence $p_{1}, \ldots, p_{m}$ with some $P_i$, for which we also use set notation.
We assume that every proposition occurs with only one sequence of parameters, i.e., there are no $a(P_i), a(P_j) \in \ap$ with $P_i \neq P_j$.

Given $\instan: P \rightarrow \params$, $P_i[\instan]$ denotes $(\instan(p_{1}), \ldots, \instan(p_{m}))$ and $a[\instan]$ denotes a($P_i[\instan]$).
For example, for $a_{\mathit{transferFrom}\_m\_n}[m \mapsto 1, n \mapsto 2]$, we obtain $a_{\mathit{transferFrom}\_1\_2}$.
We also write $\psi_\text{LTL}[\instan]$ for an LTL formula where every atomic proposition is instantiated according to $\instan$.
We define $\ap_\params = \set{a[\instan] \mid a \in \ap, \instan : P \rightarrow \params}$.
As there are no two $a(P_i), a(P_j) \in \ap$ with $P_i \neq P_j$, for any $\alpha \in \ap_\params$, there is exactly one $a$ such that $a[\instan] = \alpha$ for some $\instan$.

\section{Splitting Algorithm}
\label{sec:paramSynt}  
In the last section, we discussed that we need to approximate the parameterized pastTSL formula $\phi$ to an LTL formula $\psi_\text{LTL}$ to synthesize $\winningRegion$.
Note that $\winningRegion$ alone does not implement a strategy for $\phi$ as each parameter instance might be in a different state of $\winningRegion$ (c.f. \Cref{subsec:high-level-algorithm}).
In this section, we discuss how to split up $\winningRegion$ to enable an efficient implementation in Solidity while at the same time making sure that the generated traces realize the original formula $\phi$.

\subsection{Idea of the Algorithm}
\label{subsec:idea}
The idea of the algorithm is to split $\winningRegion$ into multiple subsystems $\winningRegion_1, \ldots, \winningRegion_n$ such that each $\winningRegion_i$ contains the transitions for method calls with parameters $P_i$.
For the ERC20 example, we produce the three systems $\winningRegion_\emptyset$, $\winningRegion_{\set{\parameter{m}}}$, and $\winningRegion_{\set{\parameter{m,n}}}$ depicted in \Cref{fig:erc20split}.
For each of these systems, at runtime, we create a copy for every instantiation of their parameters.

If a method with parameters $P_i$ is called and $\winningRegion_i$ is in state $q$, then the transition from $q$ labeled with that method call is the candidate transition to be executed.
This means that compared to the naive solution (c.f. \Cref{subsec:high-level-algorithm}) a call to \pmethod{localPause}{m=1} only has to be recorded in a single transition system (namely $\winningRegion_{\set{\parameter{m=1}}}$).

Crucially, however, we now need to ensure that we only produce traces of $\winningRegion$.
For example, if \pmethod{localPause}{m=1} is called, we move from state $q_1$ to $q_2$ in system $\winningRegion_{\set{\parameter{m=1}}}$. This corresponds to a transition from $s_1$ to $s_2$ in $\winningRegion$.
Now, for instances with \texttt{m=1}, calls to all methods except \pmethod{localUnpause}{m=1} and \method{pause} need to be rejected (according to $\winningRegion$), even though these would technically be possible in systems $\winningRegion_{\set{\parameter{m=1}, \parameter{n=}v}}$ (for any $v$).
To do so, we synchronize the systems with the help of transition guards $\texttt{in\_}s_i$ and additional state labels $\K$.

A guard $\texttt{in\_}s_i$ indicates that the transition can only be taken in state $s_i$ of $\winningRegion$.
To check if this requirement is satisfied, the systems share their knowledge about the state $\winningRegion$ would currently be in.
A knowledge label $\color{darkgreen!80}K = \set{s_1, s_2}$ in state $q$ of $\winningRegion_i$ means that $\winningRegion$ could be in state $s_1$ or state $s_2$ if $\winningRegion_i$ is in state $q$.
Each system is a projection to some transitions of $\winningRegion$ and therefore has different knowledge labels.

The systems share their knowledge in order to determine which state $\winningRegion$ would be in for a trace of one parameter instantiation.
For each method call, the systems must come to a conclusion if that call would be allowed in the current state of $\winningRegion$.
However, $\winningRegion_i$ may only use the knowledge of systems $\winningRegion_j$ with $P_j \subseteq P_i$ as these are the parameters for which there is currently a value available.
To guarantee that an unambiguous conclusion is always possible to achieve, we formulate two simple requirements and an independence check.

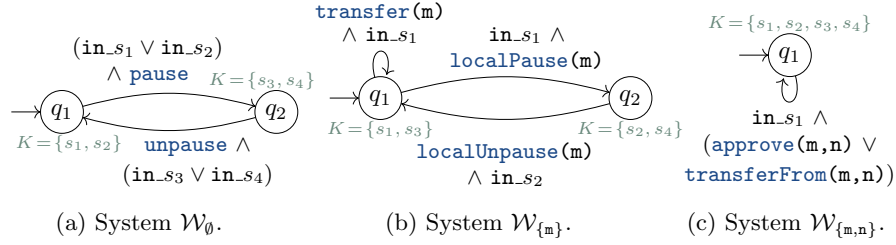
\begin{figure}[t]
	\begin{subfigure}[b]{0.33\textwidth}
		\small
		\centering
		\scalebox{0.95}{
			\begin{tikzpicture}[initial text=, 
				->,
				node distance=3cm,
				state/.style = {circle, draw, minimum size=6mm,
					inner sep=0pt, outer sep=0pt, font=\normalsize},
				state/.default = 6pt  
				]
				
				
				\node[state, initial left] (q1) {$q_1$};
				\node[state, right of=q1] (q2) {$q_2$};
				
				\node[draw=none, below = -1mm of q1, xshift=1mm] (k1) {\color{darkgreen!80}{\scriptsize$K \!=\! \{ s_1 , s_2 \}$}};
				\node[draw=none, above = -1mm of q2, xshift=-2mm] (k2) {\color{darkgreen!80}{\scriptsize$K \!=\! \{ s_3 , s_4\}$}};
				
				\draw  (q1) edge[bend left=15, above] node[xshift=-3mm]{\shortstack{$(\texttt{in\_}s_1 \lor \texttt{in\_}s_2)$ \\ $\land$ \method{pause}}} (q2)
				(q2) edge[bend left=15, below] node[xshift=4mm]{\shortstack{\method{unpause} $\land$ \\ $(\texttt{in\_}s_3 \lor \texttt{in\_}s_4)$}} (q1);	
		\end{tikzpicture}}
		\caption{System $\winningRegion_\emptyset$.}
		\label{fig:sG}
	\end{subfigure}
	\hfill
	\begin{subfigure}[b]{0.39\textwidth}
		\small
		\centering
		\scalebox{0.95}{
			\begin{tikzpicture}[initial text=, 
				->,
				node distance=3.5cm,
				state/.style = {circle, draw, minimum size=6mm,
					inner sep=0pt, outer sep=0pt, font=\normalsize}
				]
				
				
				\node[state, initial left] (q1) {$q_1$};
				\node[state, right of=q1] (q2) {$q_2$};
				
				\node[draw=none, below = -1mm of q1] (k1) {\color{darkgreen!80}{\scriptsize$K \!=\! \{ s_1 , s_3 \}$}};
				\node[draw=none, below = -1mm of q2] (k2) {\color{darkgreen!80}{\scriptsize$K \!=\! \{ s_2 , s_4\}$}};
				
				\draw  (q1) edge[bend left=15, above] node[xshift=3mm]{\shortstack{$\texttt{in\_}s_1$ $\land$ \\ \pmethod{localPause}{m}}} (q2)
				(q2) edge[bend left=15, below] node[yshift=-2mm]{\shortstack{\pmethod{localUnpause}{m} \\ $\land$ $\texttt{in\_}s_2$}} (q1)
				(q1) edge[loop above] node {\shortstack{\pmethod{transfer}{m} \\ $\land$ $\texttt{in\_}s_1$}} (q2);	
		\end{tikzpicture}}
		\caption{System $\winningRegion_{\set{\parameter{m}}}$.}
		\label{fig:sM}
	\end{subfigure}
	\hfill
	\begin{subfigure}[b]{0.26\textwidth}
		\small
		\centering
		\scalebox{0.95}{
			\begin{tikzpicture}[initial text=, 
				->,
				node distance=3cm,
				state/.style = {circle, draw, minimum size=6mm,
					inner sep=0pt, outer sep=0pt, font=\normalsize},
				state/.default = 6pt  
				]

				\node[state, initial left] (q1) {$q_1$};
				
				\node[draw=none, above = -1mm of q1] (k1) {\color{darkgreen!80}{\scriptsize$K \!=\! \{ s_1, s_2, s_3, s_4 \}$}};
				
				\draw  (q1) edge[loop below] node{\shortstack{$\texttt{in\_}s_1$ $\land$ \\ (\pmethod{approve}{m,n} $\lor$\\ \pmethod{transferFrom}{m,n})}} (q1);
		\end{tikzpicture}}
		\caption{System $\winningRegion_{\set{\parameter{m},\parameter{n}}}$.}
		\label{fig:sMN}
	\end{subfigure}
	\caption{State machines for all non-empty parameter sets. For readability, we omit cell updates and all predicates apart from method calls.}\label{fig:erc20split}
\end{figure}

\subsection{Construction}
Let $\psi_\text{LTL}$ be given.
The formula is the approximation of a \pastTSL formula and therefore ranges over $\ap = I \cup O$, where $I$ are the atomic propositions obtained from predicate terms and $O$ are the ones obtained from update terms.
For $A \subseteq \ap$, we write $A_{|O}$ instead of $A \cap O$.
We denote the set of atomic propositions that correspond to some method call \pmethod{f}{$P_i$} by $\methodProps \subseteq I$ and the set of output propositions that denote self-updates by $\selfU \subseteq O$.

Let $\winningRegion = (S_\winningRegion, s^0_\winningRegion, \delta_\winningRegion)$ be the finite-state machine over~$\ap$ that constitutes the winning region of $\psi_\text{LTL}$. $\delta_\winningRegion$ its transition relation.
We state two requirements on $\winningRegion$, which are needed to enable a sound splitting of $\winningRegion$ and can be checked easily by inspecting all its transitions.
First, we require that calls to a method parameterized with parameter sequence $P_i$ only result in cell updates parameterized with the same parameter sequence.
\begin{requirement}[Local Updates]
	\label{req:updates}
	For every transition $(s, A, s') \in \delta_\winningRegion$, if $\pfield{o}{$P_i$} \in A_{|O}$ and $\pfield{o}{$P_i$} \notin \selfU$, then there is a method call proposition $\pmethod{f}{$P_i$} \in A$.
\end{requirement}

Second, whether a method can be called at a given state must not depend on predicates with parameters that are not included in the current method call.

\begin{requirement}[Independence of Irrelevant Predicates]
	\label{req:inputs}
	For every $(s, A, s')\in \delta_\winningRegion$, if $\pmethod{f}{$P_i$} \in A$, then for every $a(P_j) \in I$ with $P_j \not\subseteq P_i$ and $a(P_j) \notin \methodProps$, there is a transition $(s, A', s')$ with $a(P_j) \in A$ iff $a(P_j) \notin A'$ and $A_{|O} = A'_{|O}$.
\end{requirement}
The above requirement is needed to unify software and state machine reasoning.
In state machines, the value of all propositions needs to be known to determine the right transition. In software, however, if \pmethod{localPause}{m} is called, the value of \texttt{n} is undefined and we cannot evaluate predicates depending on \parameter{n}.

If $\winningRegion$ satisfies the above requirements, we construct $\wini{1}, \ldots , \wini{n}$ for each parameter subset $P_i$. Each $\wini{i}$ projects $\winningRegion$ to the method calls with parameters $P_i$. 
The algorithm to construct the projections combines several standard automata-theoretic concepts:
\begin{enumerate}
	\item Introduce a new guard proposition $\texttt{in\_}s$ for every state $s \in S_\winningRegion$ of $\winningRegion$. For every transition $(s, A, s') \in \delta_\winningRegion$, replace $A$ with $A \cup \set{\texttt{in\_}s}$.
	\item Label all transitions $(s, A, s') \in \delta_\winningRegion$ for which there is no $\pmethod{f}{$P_i$} \in A$ with $\epsilon$. The result is a nondeterministic safety automaton with $\epsilon$-edges.
	\item $\wini{i}$ is obtained by determinizing the safety automaton using the standard subset construction. This removes all $\epsilon$ transitions. During the construction, we label each state with the subset of $S_\winningRegion$ it represents, these are the knowledge labels $\K$.
\end{enumerate}
We use $S_i$ for the states of $\wini{i}$, $\delta_i$ for its transition relation, and $K_i : S_i \rightarrow \pow{S_\winningRegion}$ for the knowledge labels.
Note that every transition in $\winningRegion$ is labeled with exactly one method call proposition and is therefore present in exactly one $\wini{i}$.
The following two propositions follow from the correctness of the subset construction for the determinization of finite automata.
The first proposition states that the outgoing transitions of a state $s_i \in \wini{i}$ are exactly the outgoing transitions of all states $s \in K_i(s_i)$.

\begin{proposition}
	\label{prop:knowledge2}
	For every state $s_i \in S_i$, if $s \in K_i(s_i)$, then for all $s' \in S$ and $A \subseteq \ap$, $(s, A, s') \in \delta_\winningRegion$ iff $(s_i, A \cup \set{\texttt{in\_}s}, s'_i) \in \delta_i$ for some $s_i' \in S_i$.
\end{proposition}

The second one states that the knowledge labels in $\wini{i}$ are consistent with the transitions of~$\winningRegion$.

\begin{proposition}
	\label{prop:knowledge}
	Let $(s, A, s') \in \delta_\winningRegion$ with $\pmethod{f}{$P_i$} \in A$. Then, for every state $s_i \in S_i$ with $s \in K_i(s_i)$, and every transition $(s_i, A \cup \set{\texttt{in\_}s}, s'_i) \in \delta_i$, it holds that $s' \in K_i(s'_i)$. Furthermore, for every $s_j$ of $\wini{j}$ with $i \neq j$, if $s \in K_j(s_j)$, then $s' \in K_j(s_j)$.  
\end{proposition}

\subsection{Check for Independence}
We now define the check if transitions in $\wini{i}$ can be taken independently of the current state of all $\wini{j}$ with $P_j \not\subseteq P_i$. 
If the check is positive, we can implement the system efficiently in Solidity: when a method \pmethod{f}{$P_i$} is called, we only need to update the single system $\wini{i}$ and whether the transition can be taken only depends on the available parameters.

Let $s_i$ and $s_i'$ be states in $\wini{i}$ and $A \subseteq \ap$.
Let $G_{(s_i, A, s_i')} = \set{s \mid (s_{i}, A \cup \set{\texttt{in\_}s}, s'_{i}) \in \delta_{i}}$ be the set of all guard propositions that occur on transitions from $s_{i}$ to $s'_{i}$ with $A$.
Let $P_{j_1}, \ldots P_{j_l}$ be the maximum set of parameter subsets such that $P_{j_k} \subseteq P_i$ for $1 \leq k \leq l$.
A transition $(s_i, A, s'_i)$ is \emph{independent} if for all states  $s_{j_1}, \ldots, s_{j_l}$ with $s_{j_k} \in S_{j_k}$ either 
\begin{enumerate}[label=(\roman*)]
	\item $K_i(s_i) \cap \bigcap_{1 \leq k \leq l}{K_{j_k}(s_{j_k})} \subseteq G_{(s_i, A, s_i')}$ or
	\item $(K_i(s_i) \cap \bigcap_{1 \leq k \leq l}{K_{j_k}(s_{j_k})}) \cap G_{(s_i, A, s_i')} = \emptyset$.
\end{enumerate}
The check combines the knowledge of $\wini{i}$ in state $s_i$ with the knowledge of each combination of states from $\wini{j_1}, \ldots, \wini{j_l}$.
For each potential combination, it must be possible to determine whether transition $(s_i, A, s_i')$ can be taken.
If the first condition is satisfied, then the combined knowledge leads to the definite conclusion that $\winningRegion$ is currently in a state where an $A$-transition can be taken.
If the second condition is satisfied, it definitely cannot be taken.
If none of the two is satisfied, then the combined knowledge of $P_i$ and all $P_{j_k}$ is insufficient to reach a definite answer.

Note that some state combinations $s_i, s_{j_1}, \ldots, s_{j_l}$ might be impossible to reach. But then, we have that $K_i(s_i) \cap \bigcap_{1 \leq k \leq l}{K_{j_k}(s_{j_k})} = \emptyset$ and the second condition is satisfied.
The check is successful if all transitions $(s_i, A, s'_i)$ in all $\delta_{i}$ are independent.

\subsection{Interpretation as Infinite-State Machine}
The goal of this section is to construct a state machine $\stateMachine$ from $\wini{1}, \ldots \wini{n}$ such that the original parameterized pastTSL formula $\phi$ is satisfied.
To simplify the presentation, we define $\stateMachine$ as a state machine over $\ap_\params$. Due to the direct correspondence of atomic propositions in $\ap_\params$ to predicate and update terms $\predTerms \cup \updateTerms$, a state machine for $\phi$ can easily be obtained from that.
In the following, we assume that $\winningRegion$ satisfies Requirements~\ref{req:updates} and \ref{req:inputs} and that $\wini{1}, \ldots , \wini{n}$ pass the check for independence.
We construct $\stateMachine$ as follows.

A state in $\stateMachine$ is a collection of $n = |\pow{P}|$ functions $f_1, \ldots, f_n$, where $f_i : \params^m \rightarrow S_i$ if $P_i = (p_{i_1}, \ldots, p_{i_m})$. Each $f_i$ indicates in which state of $\wini{i}$ instance $\instan$ currently is.
The initial state is the collection of functions that all map to the initial states of their respective $\wini{i}$.
For every state $s = (f_1, \ldots, f_n)$ of $\stateMachine$, every $P_i \subseteq P$, and every instance $\instan$, we add a transition where $P_i[\instan]$ takes a step and all other instances stay idle. Let
$f_i(P_i[\instan]) = s_i$, $s_i' \in S_i$, $A \subseteq \ap$,
and $G_{(s_i, A, s_i')} = \set{s \mid (s_i, A\cup\set{\texttt{in\_}s}, s_i') \in \delta_i}$.
Let $P_{j_1}, \ldots P_{j_l}$ be all subsets of $P_i$.
If $K_i(s_i) \cap \bigcap_{1 \leq k \leq l}{K_{j_k}(f_{j_k}(P_{j_k}[\instan]))} \subseteq G_{(s_i, A, s_i')}$, we add the transition $(s, A', s')$ to $\stateMachine$, where $A'$ and $s'$ are defined as follows.
\begin{align*}
	A' &= \set{a[\instan] \mid a \in A} \cup \set{o[\instan'] \mid o \in \selfU, o[\instan'] \neq o[\instan] } \\
	s' &= (f_1, \ldots, f_i[P_i[\instan] \mapsto s_i'], \ldots, f_n)
\end{align*}
The label $A'$ sets all propositions of instance $\instan$ as in $A$ and sets all other input propositions to $\false$. Of all other outputs propositions, it only sets those denoting self-updates to $\true$.

\subsection{Correctness}
Finally, we argue that $\stateMachine$ as defined above satisfies the original specification $\phi$ for all instantiations of its parameters.

\paragraph{Trace Projection.}
To obtain a compact state machine, our specifications require that in each step, exactly one method is called. Like that, the resulting specification describes the control flow projected on each instance.
To argue that $\stateMachine$ satisfies $\phi$, we therefore need to project its traces to the steps relevant for an instance $\instan$.
These are the steps that either include a method call to $\instan$ or a non-self-update of one of $\instan$'s cells.

For $A \subseteq \ap_\params$, we define $A_{\instan}$ as $\set{\alpha \in A \mid \exists a \in \ap \ldot \alpha = a[\instan]}$.
Let $\mathit{traces}(\stateMachine)$ be the set of infinite traces produced by $\stateMachine$.
Given $t \in \mathit{traces}(\stateMachine)$, let $t' = (t[0])_\instan (t[1])_\instan \ldots$.
Now, we define $t_\instan$ to be the trace obtained from $t'$ by deleting all positions $i$ such that $(t[i]_\instan)_{|O} \subseteq \selfU$ and
$\neg \exists \pmethod{f}{$P_i$} \in \methodProps \ldot \pmethod{f}{$P_i$}[\instan] \in t'[i]$.
Note that $t_\instan$ might be a finite trace even if $t$ is infinite.
Since $t_\instan$ only deletes steps from $t$ that do not change the value of the cells, $t_\instan$ still constitutes a sound computation regarding the TSL semantics.
We define $\mathit{traces}_\instan(\stateMachine) = \set{t_\instan \mid t \in \mathit{traces}(\stateMachine)}$.
\paragraph{Correctness Proof.}
Most of the work is done in the following lemma.
We define $\winningRegion_\instan$ as the state machine that replaces the transition labels of $\winningRegion$ with their instantiations according to $\instan$, i.e., if $(s, A, s') \in \winningRegion$, then $(s, A[\instan], s') \in \winningRegion_\instan$.
Not every infinite run of $\stateMachine$ corresponds to an infinite run in $\winningRegion_\instan$ for every $\instan$. 
However, we show that if the run has infinitely many $\instan$-transitions, then it can be mapped to an infinite trace in $\winningRegion_\instan$.
The proof of the lemma can be found in \Cref{a:lem:correctness}.
\begin{restatable}{lemma}{correctness}
		\label{lem:correctness}
		For every instance $\instan$, $\mathit{traces}_\instan(\stateMachine)  = \mathit{traces}(\winningRegion_\instan)$.
\end{restatable}
From the above lemma we directly obtain the desired correctness result.
\begin{theorem}
	Let $\varphi = \forall \parameter{p}_1, \ldots \parameter{p}_m \ldot \psi$ be a parameterized pastTSL formula and $\psi_\text{LTL}$ its LTL approximation.
	If $\winningRegion$ is the winning region of $\psi_\text{LTL}$, $\winningRegion$ satisfies Requirements 1 and 2, and can be split into $\wini{1}, \ldots, \wini{n}$ such that the check for independence is successful, then for every $\instan$, $\stateMachine$ defines a strategy for $\varphi[\instan]$.
\end{theorem}
\begin{proof}
	Let $\instan : P \rightarrow \params$ be an instantiation of the parameters $\parameter{p}_1, \ldots \parameter{p}_m$, $\iota \in \inputStream$ be an input stream, and $\assign$ be an assignment function.
	First, for any trace $t \in \mathit{traces}_\mu(\stateMachine)$ with $\mathit{LTL}(\iota, \mathit{comp}(t), \assign) = t$ (see \Cref{par:tsl_approx} for the definition), we have that $t \in \mathit{traces}(\winningRegion_\instan)$ because of \Cref{lem:correctness}. As all traces of $\winningRegion$ satisfy $\psi_\text{LTL}$, $t \models \psi_\text{LTL}[\instan]$ (since $\instan$ is only a renaming of atomic propositions on the LTL-level). 
	By \Cref{prop:soundness}, we obtain $\iota, \mathit{comp}(t) \models \psi[\instan]$.
	Second, as $\winningRegion$ implements the set of all strategies satisfying $\psi_\text{LTL}$, with the same reasoning, there is at least one $t$ in $\stateMachine$ with $\mathit{LTL}(\iota, \mathit{comp}(t), \assign) = t$.
\end{proof}

\subsection{Extension to Existential Quantifiers}\label{subsec:limitations}
Currently, our approach cannot handle existential quantifiers. In the example of the ERC20 contract, this forbids us to use a field \pfield{funds}{m} to store the balance of all users of the contract. If we were to try, we could use an additional parameter \texttt{r} for the recipient of the tokens and state the following.
\begin{align*}
	\forall \texttt{m}, \texttt{n}, \texttt{r} \ldot & \Globally( \pmethod{transferFrom}{m,n,r} \lor \pmethod{transfer}{m,r} \\
	& \quad \to \update{\pfield{funds}{m}}{\pfield{funds}{m} - \texttt{arg@amount}} \\
	& \qquad \land \update{\pfield{funds}{r}}{\pfield{funds}{r} + \texttt{arg@amount}})
\end{align*}
However, for completeness, we would have to specify that the \field{funds} field does not spuriously increase, which would require existential quantifiers.
\begin{align*}
	\forall \texttt{r} \ldot & \Globally( \update{\pfield{funds}{r}}{\pfield{funds}{r} + \texttt{arg@amount}} \\
	& \quad \to \exists \texttt{m} \ldot \exists \texttt{n} \ldot \pmethod{transferFrom}{m,n,r} \lor \pmethod{transfer}{m,r})
\end{align*}
A similar limitation stems from \Cref{req:updates}, which requires that a field parameterized with set $P_i$ can only be updated by a method that is also parameterized with $P_i$. As for existential quantifiers, we would otherwise not be able to distinguish spurious updates from intended updates of cells.
While it might be challenging to extend the approach with arbitrary existential quantification, it should be possible for future work to include existential quantification that prevents spurious updates.
One could, for example, define some sort of ``lazy synthesis'', which only does a non-self-update when necessary.

\section{Implementation and Evaluation}  
\label{sec:impl_eval}

\subsection{Implementation}
\label{sec:implementation}
\begin{figure*}[t]
	\centering
	\includegraphics[width=\linewidth]{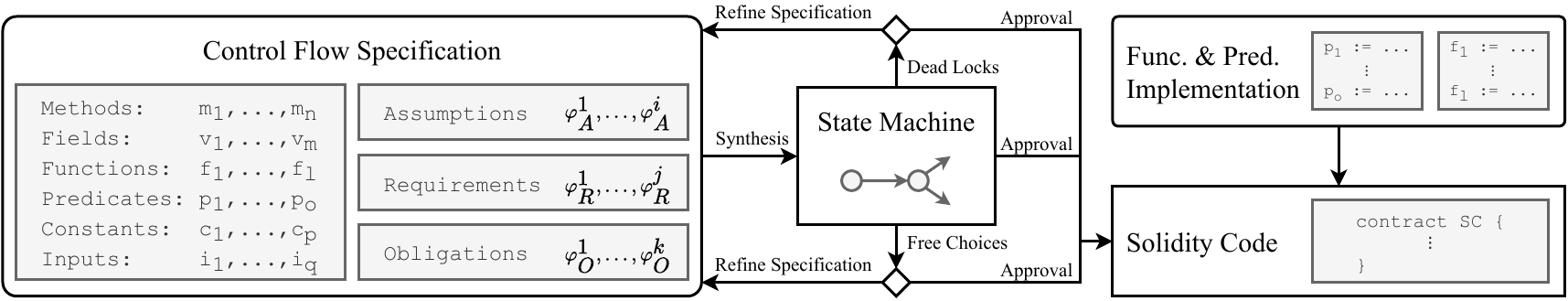}    
	\caption{Workflow of our smart contract control flow synthesis.}
	\label{fig:workflow}
\end{figure*}

We implemented our approach in a toolchain consisting of several steps.
First, we translate the \pastTSL specification into a \pastLTL formula using TSLtools~\cite{tsltools}, which we adapted to handle past-time operators.
We then synthesize a state machine using BDD-based symbolic synthesis.
To make our lives easier, we implemented a simple analysis to detect free choices and deadlocks, which both indicate potential specification errors.
If the specification contains parameters, we split the resulting state machine as described in \Cref{sec:paramSynt}.
Lastly, the state machines are translated to Solidity code. 
The toolchain is implemented in our tool \tool consisting of approximately 3000 lines of Python code (excluding TSLtools).
From a user perspective, we obtain the workflow depicted in~\Cref{fig:workflow}.

\paragraph{Synthesis from PastLTL.}
The first part of our toolchain implements a symbolic synthesis algorithm for \pastLTL.
As such, it can also be employed outside the context of smart contract synthesis.
We are not aware of any other tool that implements pastLTL synthesis.
We first build the safety automaton of the specification using a representation as BDDs. 
For \pastLTL, a symbolic approach is especially efficient due to the long-known fact that for evaluating a \pastLTL formula at time point $i$, it is sufficient to know the value of all subformulas at point $i-1$~\cite{recursive_evaluation}.
Afterwards, we symbolically extract the winning region from the safety region with a classic fixpoint attractor construction.
Finally, we minimize the resulting state machine using an explicit implementation of Hopcroft's minimization algorithm~\cite{hopcroft1971n}.

\paragraph{State Machine Analysis.}
We analyze the winning region for free choices and potential deadlocks, which both usually indicate specification errors.
A free choice is a state which, for the same input, has multiple outgoing transitions into the winning region.
If there are free choices and the developer has no preference which one is chosen, \tool nondeterministically commits to one option.
For the deadlock detection, we require the user to label \emph{determined predicate terms}. We call a predicate determined if either 1) it becomes a constant at some time or 2) only method calls can change its value.
An example of class 1 are predicates over the time, e.g., $\texttt{time} > \texttt{cTime()}$: if it is $\true$ at some point, it will never be $\false$ again.
A class 2 example would be a predicate that counts whether the number stored in a field has passed a fixed threshold.
A predicate like $\texttt{msg.sender} = \texttt{owner()}$, on the other hand, is not determined as the evaluation changes with the input \texttt{msg.sender}.
\tool au\-to\-matically detects if, at some point, there is an evaluation of the determined predicate terms that is allowed by the assumptions but for which there is no valid transition.
It then warns of a potential deadlock.

\paragraph{Translation to Solidity.}\label{subsec:dot_to_sol}
For the translation, the developer needs to provide the implementation of all predicates and functions, as they are uninterpreted (which makes the synthesis feasible after all).
Some of the most common functions and predicates (e.g., equality and addition) are automatically replaced by \tool. 
The \texttt{owner} and \texttt{msg.sender} keywords are translated automatically; the owner is set in the constructor.
Conceptually, the translation to Solidity is straightforward.
For each method of the contract, we create a function that contains the state machine logic for that particular method.
For parameterized specifications, the contract is augmented with a mapping recording the knowledge labels (c.f. \Cref{sec:paramSynt}).
The parameters other than the sender are included as arguments.
Following~\cite{DBLP:conf/fc/MavridouL18}, we also add automatic protection against reentrancy attacks by setting a Boolean flag if a method is currently executing.

\subsection{Evaluation}
\label{sec:eval}
\label{subsec:benchmarks}
The goal of our evaluation is to show that 1) parameterized pastTSL is indeed a suitable logic for specifying smart contract state machines and 2) that our implemented toolchain is efficient.
To do so, we specified and synthesized ten different smart contracts with a non-trivial temporal control flow using pastTSL specifications with and without parameters.
A detailed description of all benchmarks is provided in \Cref{a:benchmarks}.
The most challenging benchmark to specify was the NFT auction, a parameterized specification for a contract actively maintained by Avolabs. Its reference implementation has over 1400 lines of code. We manually extracted 30 past-time formulas from the README of the contract provided on the GitHub of Avolabs~\cite{NFTReference}.

All experiments were run on a 2020 Macbook with an Apple M1 chip, 16GB RAM, running MacOS.
The results are shown in \Cref{tab:results}.
We report the size of the specification and of the resulting state machine as well as the running time of the synthesis procedure itself and the translation to Solidity code.
\begin{table*}[t]
	\caption{Sizes of the specifications and state machines as well as the average running time of \tool.
	\#Forms. is the number of individual past-time formulas, \#Nodes is the number of nodes of the AST. The state machine size is the sum of the states/transitions of the split state machines. The synthesis and translation times are the respective average on 10 runs of the same benchmark.}
	\label{tab:results}
	\small
	\def\arraystretch{0.9}
	\setlength{\tabcolsep}{3pt}
	\centering
	\begin{tabular}{r||>{\centering}p{1.2cm}|>{\centering}p{1.2cm}|>{\centering}p{1.2cm}|>{\centering}p{1.2cm}|>{\centering}p{1.5cm}|>{\centering\arraybackslash}p{1.5cm}}
		& \multicolumn{2}{c|}{Specification} & \multicolumn{2}{c|}{State Machine} & \multicolumn{2}{c}{Avg. Time (s)} \\
		Contract & \footnotesize{\#Forms.} & \footnotesize{\#Nodes} & \footnotesize{\#States} & \footnotesize{\#Trans.} & \footnotesize{Synth.} & \footnotesize{Transl.} \\
		\hline\hline
		Asset Transfer & 36 & 216 & 8 & 14 & 5.9996 & 0.0053 \\
		Blinded Auction & 19 & 218 & 5 & 8 & 1.5446 & 0.0026 \\
		Coin Toss & 27 & 154 & 5 & 7 & 1.6180 & 0.0029 \\
		Crowd Funding & 17 & 100 & 4 & 8 & 0.2178 & 0.0026 \\
		ERC20 & 15 & 140 & 9 & 5 & 0.4812 & 0.0033 \\
		ERC20 Extended & 19 & 244 & 10 & 7 & 1.9608 & 0.0040 \\
		NFT Auction & 30 & 325 & 12 & 15 & 12.1853 & 0.0080 \\
		Simple Auction & 15 & 83 & 4 & 7 & 0.1362 & 0.0026 \\
		Ticket System & 13 & 97 & 4 & 6 & 0.1812 & 0.0028 \\
		Voting & 17 & 98 & 6 & 5 & 0.1478 & 0.0023
	\end{tabular}
\end{table*}
Most importantly, the evaluation shows that specifying and automatically generating the non-trivial state machine logic of a smart contract is possible.
We successfully synthesized Solidity code for state machines of up to 12 states.
The evaluation also shows that our toolchain is efficient: synthesis itself took up to 12 seconds; in most cases, \tool synthesizes a state machine in less than two seconds.
The translation of the state machine into Solidity code is nearly instantaneous.

\section{Conclusion}
We have described the synthesis of Solidity code from specifications given in \pastTSL equipped with universally quantified parameters.
Our approach is the first that facilitates a comprehensive specification of the implicit state machine of a smart contract, including the data flow of the contract's fields and guards on the methods' arguments.
The algorithm proceeds in two steps: first, we translate the specification to pastTSL.
While we have shown that \pastTSL realizability without parameters is undecidable in general, solutions can be obtained via a sound reduction to LTL.
In a second step, we split the resulting system into a hierarchical structure of multiple systems, which constitutes a finite representation of a system implementing the original formula and also enables a feasible handling when translated to Solidity.
Our prototype tool \tool implements the synthesis toolchain, including an analysis of the state machine regarding potential specification errors. 

For future work, we aim to extend our approach to specifications given in \pastTSL with alternating parameter quantifiers.
There are also several exciting possibilities to combine our work with other synthesis and verification techniques. One avenue is to automatically prove necessary assumptions in deductive verification tools~\cite{dharanikota2021celestial}, especially for assumptions that state invariants maintained by method calls.
Another possibility is to synthesize function and predicate implementations in the spirit of~\cite{cispa3674}.
Finally, now that we have developed the algorithmic foundations and implemented a first prototype, we aim to conduct a thorough evaluation of our approach in comparison to hand-written (non-formal) approaches.

%
%
%
 \bibliographystyle{splncs04}
 \bibliography{bibliography}

\appendix
\label{a:prelims}

\section{Proof of Lemma \ref{lem:correctness}}
\label{a:lem:correctness}
\correctness*
	In the following, we show that $\mathit{traces}_\instan(\stateMachine) \cup \mathit{finTraces}_\instan(\stateMachine) = \mathit{traces}(\winningRegion_\instan) \cup \mathit{finTraces}(\winningRegion_\instan)$, where $\mathit{finTraces(\cdot)}$ and $\mathit{finTraces}_\instan(\cdot)$ are the sets of finite prefixes of the traces in $\mathit{traces(\cdot)}$ and $\mathit{traces}_\instan(\cdot)$, respectively.
	From this, the original statement of the lemma follows.
	We therefore have to show that every (finite or infinite) run $r_\stateMachine$ of $\stateMachine$ can be matched with a (finite or infinite) run $r_{\winningRegion_\instan}$ of $\winningRegion_\instan$ (and vice versa) such that $\mathit{trace}_\instan(r_\stateMachine) = \mathit{trace}(r_{\winningRegion_\instan})$.
	For the first direction of the equality, we show that every transition of $r_\stateMachine$ can be either matched with a transition in $\winningRegion_\instan$ or constitutes a step which is removed in the trace $t_\instan = \mathit{trace}(r_\stateMachine)_\instan$.
	For the other direction, we show that every transition of $r_{\winningRegion_\instan}$ can be matched with a transition in $\stateMachine$.
	Assume $\stateMachine$ is currently in state $s = (f_1, \ldots, f_n)$ of $r_\stateMachine$ and $\winningRegion_\instan$ is in state $s_\winningRegion$ of $r_{\winningRegion_\instan}$.
	We keep the invariant that $s_\winningRegion \in K_\instan = \bigcap_{1 \leq j \leq n}{K_j(f_j(P_j[\instan]))}$.
	\begin{itemize}
		\item
		For the first direction, assume that the next transition of $\stateMachine$ is $(s, A, s')$ with $s' = (f'_1, \ldots, f'_n)$. Let $B \subseteq \ap$ be such that $B[\instan] = A_{\instan}$.
		By construction of $\stateMachine$, there is an instance $\instan'$ and parameter subset $P_i$ such that there is exactly one $\pmethod{f}{$P_i[\instan']$} \in A$.
		Let $B' \subseteq \ap$ be such that $B'[\instan'] = A_{\instan'}$.
		We distinguish two cases.
		\begin{itemize}
			\item Assume $P_i[\instan] = P_i[\instan']$.
			Let $s_i = f_i(P_i[\instan'])$ and $s'_i = f'_i(P_i[\instan'])$.
			Let $G_{(s_i, B', s_i')} = \set{s \mid (s_i, B' \cup\set{\texttt{in\_}s}, s_i') \in \delta_i}$.
			Let $P_{j_1}, \ldots P_{j_l}$ be all subsets of $P_i$ and
			$K = K_i(s_i) \cap \bigcap_{1 \leq k \leq l}{K_{j_k}(f_{j_k}(P_{j_k}[\instan']))}$.
			By definition of the independence check employed to construct $\stateMachine$, $K \subseteq G_{(s_i, B', s_i')}$.
			As $P_i[\instan] = P_i[\instan']$, we have that $K_\instan \subseteq K$ and thus also $K_\instan \subseteq G_{(s_i, B', s_i')}$. Therefore, by the invariant, $s_\winningRegion \in G_{(s_i, B', s_i')}$ and there is a transition $(s_\winningRegion, B', s'_\winningRegion)$ in $\winningRegion$.
			
			Now, if $\instan' = \instan$, there is a transition $(s_\winningRegion, B[\instan], s'_\winningRegion)$ in $\winningRegion_\instan$, what shows that the step in $\stateMachine$ can be mirrored in $\winningRegion_\instan$.
			Otherwise, if $\instan' \neq \instan$, $B$ and $B'$ agree on the propositions parameterized with $P_i$.
			By construction of $\stateMachine$, $B$ has all inputs not parameterized with $P_i$ set to $\false$ and only self-updates of cells not parameterized with $P_i$ set to $\true$. 
			By Requirements~\ref{req:updates} and \ref{req:inputs}, since there is a transition $(s_\winningRegion, B', s'_\winningRegion)$ in $\winningRegion$, there is also a transition $(s_\winningRegion, B, s'_\winningRegion)$ in $\winningRegion$ and therefore a transition $(s_\winningRegion, B[\instan], s'_\winningRegion)$ in $\winningRegion_\instan$.
			
			All that remains to show is that the invariant is preserved by both transitions $(s_\winningRegion, B[\instan], s'_\winningRegion)$ and $(s_\winningRegion, B'[\instan'], s'_\winningRegion)$. That means that we need to show that $s'_\winningRegion \in K'_\instan$ for $ K'_\instan = \bigcap_{1 \leq j \leq n}{K_j(f'_j(P_j[\instan]))}$ in either case.
			Since $s_\winningRegion \in K_\instan$, we have that for every $j$, $s_\winningRegion \in K_j(f_j(P_j[\instan]))$. 
			Furthermore, $\pmethod{f}{$P_i$} \in B$ and $\pmethod{f}{$P_i$} \in B'$. 
			For $j \neq i$ we have that $f'_j = f_j$ and, therefore, by \Cref{prop:knowledge}, $s'_\winningRegion \in K_j(f'_j(P_j[\instan]))$. 
			Furthermore, $K_\instan \subseteq G_{(s_i, B', s_i')}$ and $(s_i, C \cup \set{\texttt{in\_}{s_\winningRegion}}, s_i') \in \delta_i$, for both $C = B$ and $C = B'$. 
			Therefore, again by \Cref{prop:knowledge}, $s'_\winningRegion \in f_i[P_i[\instan] \mapsto s'_i]$. Thus, $s'_\winningRegion \in K_j(f_j(P_j[\instan]))$ for all $j$ and the invariant holds for $s'_\winningRegion$.
			
			\item Assume $P_i[\instan] \neq P_i[\instan']$. By construction of $\stateMachine$, there is no $\pmethod{g}{$P_j$} \in B$.
			As $\winningRegion$ satisfies \Cref{req:updates}, so does every $\wini{i}$. Together with how $\stateMachine$ is constructed, it follows that $B_{|O} \subseteq \selfU$.
			Thus, by definition of $t_\instan$, the transition will not appear in $t_\instan$ and $\winningRegion_\instan$ thus stays in its current state.
			In remains to show that the invariant is maintained by the transition in $\stateMachine$. 
			Since $P_i[\instan] \neq P_i[\instan']$, we have by definition of $\stateMachine$ that $f_j(P_j[\instan]) = f'_j(P_j[\instan])$ for all $j$, and therefore the invariant is maintained.
		\end{itemize}
		
		\item For the second direction, assume that the next transition of $\winningRegion_\instan$ is $(s_\winningRegion, B[\instan], s'_\winningRegion)$ for some $B \subseteq \ap$.
		By construction of $\winningRegion$, there is a parameter subset $P_i$ such that some $\pmethod{f}{$P_i$} \in B$.
		By our invariant, $s_\winningRegion \in K_i(s_i)$, 
		therefore, by \Cref{prop:knowledge2}, there is a transition $(s_i, B \cup \set{\texttt{in\_}s_\winningRegion}, s'_i) \in \delta_i$.
		By construction of $\stateMachine$, there is a transition from $(f_1, \ldots, f_n)$ to $(f'_1, \ldots, f'_n)$ such that $f'_i(P_i[\instan]) = s'_i$. Furthermore, the transition is labeled with $A$ such that $A_{\instan} = B[\instan]$.
		
		It remains to show that the invariant is maintained by the transition. For $j \neq i$, $f'_j = f_j$ and therefore, by \Cref{prop:knowledge}, $s'_\winningRegion \in K_j(f_j(P_j[\instan]))$. Furthermore, since $(s_i, B \cup \set{\texttt{in\_}s_\winningRegion}, s'_i) \in \delta_i$, again by \Cref{prop:knowledge}, it holds that $s'_\winningRegion \in K_i(s'_i)$ and therefore, by definition of $f'_i = f_i[P_i[\instan] \mapsto s'_i]$, $s'_\winningRegion \in K_i(f'_i(P_i[\instan]))$. Therefore, the invariant is maintained for the transition.\qedhere
	\end{itemize}

\section{Description of Benchmarks}
\label{a:benchmarks}
\begin{itemize}[leftmargin=20pt]
	\item \emph{Asset Transfer} (not parameterized, reference contract: \cite{AssetTransferReference}). This contract originates from Microsoft's Azure Blockchain Workbench~\cite{Azure}. It has been formulated similarly in~\cite{DBLP:conf/vstte/0001LCPDBNF19,dharanikota2021celestial}. 
	\item \emph{Blinded Auction} (not parameterized, reference contract: \cite{BlindedAuctionReference}). The blinded auction protocol, following~\cite{DBLP:conf/fc/MavridouL18}, describes how to sell an item to the highest bidder. The bids of all bidders are hashed and are only revealed after the auction closed.
	\item \emph{Coin Toss} (not parameterized, reference contract: \cite{CoinTossReference}). In this protocol, following~\cite{Azure}, two users bet, together with a wager, on whether a coin toss results in heads or tails. After reaching a time limit, the winner receives both wagers.
	\item \emph{Crowd Funding} (not parameterized, reference contract: \cite{CrowdFundingReference}). In this contract, following~\cite{DBLP:journals/corr/abs-1906-02906}, users can donate coins until a time limit is reached. If the funding goal is reached by the end, the owner retrieves the donations. Otherwise, the users can reclaim their donations.
	\item \emph{ERC20 Token System} (parameterized, reference contract: \cite{ERC20Reference}). We consider a classical ERC20 token system following the Open Zeppelin documentation~\cite{ERC20Spec} but without pausing feature.
	\item \emph{ERC20 Token System Extended} (parameterized). This is the contract described in this paper which extends the first ERC20 token system with different methods for pausing.
	\item \emph{NFT Auction} (parameterized, reference contract: \cite{NFTReference}). This is a comparatively large contract from Avolabs that implements an NFT auction combined with a buy-now feature. The original implementation has over 1400 lines of code. We specified the main requirements on the control flow as described in the \texttt{README} of Avolabs' GitHub.
	\item \emph{Simple Auction} (not parameterized, reference contract: \cite{SimpleAuctionReference}). The simple auction protocol is a common example for smart contract control flows, see, e.g.,~\cite{SimpleAuction}. It is similar to the blinded auction without the bids being hashed.
	\item \emph{Ticket System} (not parameterized, reference contract written by the authors of this paper). The ticket system contract describes the selling process of tickets. As long as tickets are available and the sale has not been closed, users can buy tickets. Users can return their tickets while the ticket sale is still open. 
	\item \emph{Voting} (parameterized, reference contract: \cite{VotingReference}). This is a simple voting contract which enforces that every voter vote only votes once. Additionally, we let the owner close the contract only after a fixed number of votes is reached.
\end{itemize}

\end{document}